\def\draft{1}
\def\llncs{0}

\def\anon{0}
\def\notoc{0}

\ifnum\draft=0

\fi

\ifnum\llncs=0
\documentclass[11pt]{article}
\usepackage{fullpage}
\else
\documentclass[orivec,runningheads,envcountsame,envcountsect]{llncs}
\usepackage{amssymb,amsmath,cite,url}
\fi

\usepackage{bm,bbm}

\newcommand{\remove}[1]{}
\newcommand{\ignore}[1]{}
\ifnum\llncs=0
\usepackage{amsthm}
\else
\fi

\ifnum\draft=1
\def\ShowAuthNotes{1}
\else
\def\ShowAuthNotes{0}
\fi

\usepackage{mathtools}
\usepackage{amsmath,amssymb}
\usepackage{url}
\usepackage{cite}

\usepackage[dvipsnames]{xcolor}
\definecolor{DarkBlue}{RGB}{0,0,150}
\definecolor{llg}{gray}{0.95}
\definecolor{lg}{gray}{0.85}

\usepackage[colorlinks,linkcolor=DarkBlue,citecolor=DarkBlue]{hyperref}

\ifnum\llncs=0

\newtheorem{theorem}{Theorem}[section]
\newtheorem{proposition}[theorem]{Proposition}
\newtheorem{definition}[theorem]{Definition}

\newtheorem{claim}{Claim}[theorem]
\newtheorem{tlclaim}[theorem]{Claim}
\newtheorem{lemma}[theorem]{Lemma}

\newtheorem{corollary}[theorem]{Corollary}

\else

\spnewtheorem{prop}{Property}{\bfseries}{\itshape}
\spnewtheorem{fact}{Fact}{\bfseries}{\itshape}
\spnewtheorem{subclaim}{Claim}[theorem]{\bfseries}{\itshape}
\spnewtheorem{tlclaim}[theorem]{Claim}{\bfseries}{\itshape}
\spnewtheorem{assumption}{Assumption}{\bfseries}{\itshape}
\fi

\ifnum\llncs=1
\def\pfend{\hfill\qedsymbol}
\else
\def\pfend{}
\fi

\newenvironment{claimproof}{\begin{proof}

}{\pfend\end{proof}}

\def\cA{{\cal A}}

\def\cK{{\cal K}}

\def\cU{{\cal U}}

\def\bbC{{\mathbb C}}
\def\bbE{{\mathbb E}}

\def\bbN{{\mathbb N}}

\newcommand{\floor}[1]{\left\lfloor #1 \right\rfloor}

\def\binset{\{0,1\}}
\def\pmset{\{\pm 1\}}

\def\q2{\lfloor q/2 \rceil}

\newcommand{\abs}[1]{\left\vert {#1} \right\vert}
\newcommand{\norm}[1]{\left\| {#1} \right\|}

\def\poly{{\rm poly}}

\def\negl{{\rm negl}}

\newcommand{\adv}{\mathrm{Adv}}

\newcommand{\mx}[1]{\mathbf{{#1}}}

\newcommand{\Ex}{\mathop{\bbE}}

\newcommand{\boldpar}[1]{\vspace{3pt}\par\noindent\textbf{#1}}

\ifnum\ShowAuthNotes=1
\newcommand{\authnote}[3]{\textcolor{#3}{[{\footnotesize {\bf #1:} { {#2}}}]}}
\else
\newcommand{\authnote}[3]{}
\fi

\newcommand{\absnewline}{\ifnum\llncs=1 \\ \fi}

\def\mgp[#1]{\mx{G}_{#1}}

\def\mgip[#1]{\mx{G}_{#1}^{-1}}

\def\mcit[#1]{\mci[#1]^T}
\def\mci[#1]{\mx{C}_{#1}}

\newcounter{hybridcount}
\newcounter{prevhybridcount}
\newcounter{nexthybridcount}

\newcommand{\ket}[1]{|{#1}\rangle}
\newcommand{\bra}[1]{\langle{#1}|}
\newcommand{\ketbra}[1]{\ket{{#1}}\bra{{#1}}}
\newcommand{\braket}[1]{{\langle {#1} \rangle}}

\def\kalpha{\ket{\alpha}}

\newcommand{\tr}{\mathop{\mathrm{Tr}}}

\newcommand{\TD}{\mathrm{TD}}

\newcommand{\magn}{\mathsf{Mag}}

\newcommand{\csv}{\mathsf{cv}}
\newcommand{\ncsv}{\overline{\csv}}

\newcommand{\cnum}{\abs{\csv}}

\newcommand{\edgep}{\mathsf{p}}

\def\piu{\Pi^{\star}}

\def\zon{{\binset^{n}}}

\def\adv{{\mathcal{A}}}

\newcommand{\uuntt}[1]{\cU^{\star}_{n,{#1}t}}
\def\unist{{\text{\sf uni}_{s,t}}}
\def\Nst{{N^{\underline{st}}}}

\newcommand{\alsup}[1]{\alpha^{({#1})}}
\newcommand{\alssup}[1]{{\alpha^*}^{({#1})}}

\newcommand{\vxsup}[1]{{x}{({#1})}}

\title{Real-Valued Somewhat-Pseudorandom Unitaries %

\author{
Zvika Brakerski
\thanks{Weizmann Institute of Science. Email: \texttt{\{zvika.brakerski,nir.magrafta\}@weizmann.ac.il}.
Supported by the Israel Science Foundation (Grant No.\ 3426/21), and by the Horizon Europe Research and Innovation Program via ERC Project ACQUA (Grant 101087742).}
\and
Nir Magrafta\addtocounter{footnote}{-1}\footnotemark
}
\date{}
\ifnum\llncs=1
\ifnum\anon=0
\thanks{For the most up-to-date version of this work, please refer to \url{https://arxiv.org/}.}
\fi
\fi
}

\ifnum\llncs=1
\author{Zvika Brakerski\thanks{%
		Supported by the Israel Science Foundation (Grant No.\ 3426/21), and by the Horizon Europe Research and Innovation Program via ERC Project ACQUA (Grant 101087742).}
}
\institute{Weizmann Institute of Science, Israel}
\date{}
\else
\ifnum\anon=1
\author{}
\else
\author{Zvika Brakerski\thanks{Weizmann Institute of Science, Israel, \texttt{zvika.brakerski@weizmann.ac.il}. Supported by the Israel Science Foundation (Grant No.\ 3426/21), and by the European Union Horizon Europe Research and Innovation Program via ERC Project ACQUA (Grant 101087742).}
}
\fi

\ifnum\draft=1
\date{\today}
\else
\date{}
\fi

\fi

\ifnum\llncs=1
\renewcommand{\paragraph}{\boldpar}
\fi

\begin{document}

\maketitle

\begin{abstract}
	
We explore a very simple distribution of unitaries: random (binary) phase -- Hadamard -- random (binary) phase -- random computational-basis permutation.
We show that this distribution is statistically indistinguishable from random Haar unitaries for any polynomial set of orthogonal input states (in any basis) with polynomial multiplicity.
This shows that even though real-valued unitaries cannot be completely pseudorandom (Haug, Bharti, Koh, arXiv:2306.11677), we can still obtain some pseudorandom properties without giving up on the simplicity of a real-valued unitary.

Our analysis shows that an even simpler construction: applying a random (binary) phase followed by a random computational-basis permutation, would suffice, assuming that the input is orthogonal and \emph{flat} (that is, has high min-entropy when measured in the computational basis).

Using quantum-secure one-way functions (which imply quantum-secure pseudorandom functions and permutations), we obtain an efficient cryptographic instantiation of the above.

\end{abstract}

\ifnum\llncs=0
\ifnum\notoc=0

\newpage
\tableofcontents

\newpage
\fi
\fi

\section{Introduction} 
\label{section:introduction}
Pseudorandomness is one of the most fundamental notions in cryptography. Prominent examples include pseudorandom generators (PRG) \cite{haastad1999pseudorandom}, pseudorandom functions (PRF) \cite{goldreich1986construct_prf}, and pseudorandom permutations (PRP) \cite{luby1988construct_prp}, which play a crucial role in various constructions in cryptography and beyond. Let us consider the concept of PRP, which is quite analogous to the object at the focus of this work. If we consider the class of all permutations $\binset^n \to \binset^n$, a random function from this class requires an exponential number of random bits to specify, and requires an exponential-size circuit to evaluate (and invert). A PRP is a distribution that can be sampled using a \emph{polynomial} number of bits, known as the \emph{seed}.\footnote{In a formal definition, one has to address the subtlety of whether ``efficiency'' is defined with respect to the input length $n$, or with respect to some ``security parameter''. This distinction does not matter for our current discussion, and we will point out when it does down the line.} Furthermore, given the seed $s$, it is possible to evaluate and invert the associated permutation $\pi_s$ in polynomial time. The crucial point is that any polynomial time process cannot distinguish between an interaction with a permutation $\pi_s$ for a random seed $s$, and an interaction with a completely random permutation. It has been established \cite{haastad1999pseudorandom,goldreich1986construct_prf, luby1988construct_prp} that PRG, PRF and PRP can all be constructed given the existence of one-way functions: the most basic (classical) cryptographic primitive. Furthermore, this connection is true even if we consider a quantum (polynomial-time) adversary \cite{zhandry2012construct,zhandry2016note}.

In the context of quantum computing and quantum cryptography, Ji Liu and Song \cite{ji2018pseudorandom} (henceforth JLS) proposed to study pseudorandomness of quantum objects. In particular, the defined the notion of \emph{pseudorandom quantum states} (PRS) which are $n$-qubit states which are (indistinguishable from) Haar random $n$-qubit states, even with arbitrary polynomial-time interaction with a polynomial number of copies of the pseudorandom state. The notion of PRS has been the subject of extensive study since \cite{brakerski2019pseudo, brakerski2020scalable,brakerski2022computational ,ananth23prfs, behera2023pseudorandomness, giurgica2023pseudorandomness, jeronimo2023subset}.

Another object proposed by JLS is that of pseudorandom unitaries (PRU).
Similarly to PRP, these should allow to evaluate and invert a unitary given a seed of polynomially many random bits,
while being computationally indistinguishable from a random Haar unitary given arbitrary polynomial time interaction. Contrary to PRS, JLS presented constructions of a PRU, but were unable to prove their security.
To date, it is still unknown how to construct PRU with a security proof under any known cryptographic assumption.

Recently, some partial progress has been made towards a construction of PRU. Namely, several works introduced families of unitaries with polynomial seeds and efficient evaluation, but falling short on pdeudorandomness. Lu, Qin, Song, Yao, and Zhao \cite{lu2023prss} introduced the notion of Pseudorandom State Scramblers (PRSS), which are unitaries that are only proven to act pseudorandomly on an arbitrary \emph{single input state} with arbitrary polynomial multiplicity. Namely, any number of copies of the output on one particular input are pseudorandom. Ananth, Gulati, Kaleoglu, and Lin \cite{ananth2023pseudorandom_isometries} introduced the notion of Pseudorandom Isometries (PRI), which are not unitary since their output is longer than their input (and the security of the constructions hinges on this property). They proved security of the PRI property for the case of a single input, a polynomial number of Haar-random inputs, or for an inputs which are a subset of computational basis elements (all with arbitrary polynomial multiplicity). The works of \cite{lu2023prss,ananth2023pseudorandom_isometries} mention a number of applications for the primitives that they defined, including multi-copy security for quantum public-key encryption.

Interestingly, the constructions in \cite{lu2023prss,ananth2023pseudorandom_isometries} are \emph{real-valued}. Namely, the unitary family consists of unitaries with only real values.\footnote{In \cite{lu2023prss} there is an additional construction which uses complex unitaries, but for the PRSS property, a real valued construction suffices.} In contrast, Haug, Bharti, and Koh \cite{haug2023pseudorandom} showed that \emph{full PRU security} cannot be achieved in this way. This is done by observing that if $U$ is real valued, then $U\otimes U$ acts as identity on the maximally entangled state, whereas this is very far from being the case if $U$ is a random unitary. Therefore, if we consider adversaries that are allowed to make entangled queries to the unitary, then it is impossible to construct real-valued pseudorandom unitaries. Indeed, very recently, and concurrently and independently of our work, Metger, Poremba, Sinha, and Yuen \cite{metger2024pseudorandom} showed that it is possible to construct \emph{non-adaptive PRU} that are pseudorandom with respect to any set of inputs that cannot change adaptively throughout the querying process, even when those inputs are entangled with each other and/or the environment. Indeed, their construction is inherently complex-valued. The question, therefore, remains:
\begin{center}
	\emph{What pseudorandom properties can be shown for real-valued efficiently computable unitaries?}
\end{center}
In this work, we show that it is possible to achieve stronger security notions than \cite{lu2023prss,ananth2023pseudorandom_isometries} using an extremely simple construction.

\subsection{Our Results}
We consider an extremely simple family of unitaries: $U_{P} U_{G} H^{\otimes n}  U_{F}$, where for functions $F,G: \binset^n \to \{\pm 1\}$, the operators $U_{F}, U_{G}$ are the unitaries $\ket{x} \to {F(x)}\ket{x}, \ket{x} \to {G(x)}\ket{x}$, and for a permutation $P: \binset^n \to \binset^n$, $U_{P}$ is the unitary $\ket{x} \to \ket{P(x)}$. The functions $F,G$ in our construction are quantum-secure pseudorandom functions, and the permutation $P$ is a quantum-secure pseudorandom permutation. 
Using the security of $F,G$ and $P$, we can replace them with truly random counterparts $f, g$ and $\pi$. We then analyze the output of the construction information theoretically with $f,g$ and $\pi$.
We note that this construction is in the spirit of, but simpler than, the PRU candidates considered by \cite{ji2018pseudorandom}.

We show that our family of unitaries acts as a PRU so long as 
the inputs are (a mixture of) an \emph{orthogonal} set of  quantum states, with arbitrary polynomial multiplicity each. This in particular shows that this construction is also a PRSS. Our construction also generalizes the properties proven by \cite{ananth2023pseudorandom_isometries} for PRI, without increasing the output size and using a construction of comparable complexity.

Notably, our construction can be separated into two parts, each of which is interesting in its own right. First, we show that $H^{\otimes n}  U_{f}$ is a ``state-flattener'', in the sense that for any polynomial size set of input states, it holds that with overwhelming probability over a truly random function $f$, the output states are all ``almost perfectly flat'' in the computational basis. Namely, the square-magnitude of each computational basis element is bounded by $\epsilon = O(\frac{n}{2^n})$ (note that $1/2^n$ is the maximum possible flatness, and Haar random states are also expected to have $\sim \frac{n}{2^n}$ flatness). This property follows immediately from known concentration bounds, but we believe that it was not explicitly pointed out in this context. So, if we only want to approximate the flattening property of PRU, it can be done almost trivially.

We then show that the second part of our construction, $U_{\pi} U_{g}$ for random $\pi, g$, acts as PRU for flat orthogonal inputs. Again, this is an extremely simple construction that can be applied even as-is for a non-trivial set of input states (e.g.\ some polynomial subset of a random basis for the given Hilbert space). The technical crux of our paper is in the analysis of this component.

\subsection{Technical Overview} \label{subsection:technical_overview}

We provide an overview of the proof for our main information theoretic lemma. That is, taking $f,g,\pi$ to be random functions and a permutation, then our construction is statistically indistinguishable from a Haar random unitary for orthogonal inputs.

We use an (approximate) characterization of the output of querying a Haar random unitary on orthogonal input states. For $t$ copies of $s$ different orthogonal inputs, the output ``target'' state can be approximated by the density matrix
\begin{align}\label{eq:intro:targetstate}
	\rho_{\text{target}}=\sum_{{z , \sigma}} \ket{z} \bra{\sigma(z)}~,
\end{align}
ignoring global normalization. The summation is over all $z \in (\binset^n)^{st}$ whose $st$ entries are unique elements in $\zon$, and over a set of permutations $\sigma$ over the set $[st]$ (or, equivalently, over $[t]^s$). Namely, the permutation $\sigma$ takes a vector $z \in (\binset^n)^{st}$ and permutes its entries (the vector $\sigma(z)$ has the same set of entries as $z$, only in a different order). Specifically, the summation is only over ``block-preserving'' permutations, which are permutations that only swap elements inside each $t$-tuples of elements. That is, a permutation is block-preserving if it can be represented as a sequence of $s$ permutations over $[t]$. %

We then prove that the output of our construction, with random functions and permutation, is close to the state in Eq.~\eqref{eq:intro:targetstate}, thereby showing that on our set of input states, our construction is statistically indistinguishable from a Haar random unitary.

In this overview we first explain how to prove the flatness property for the first part of our construction, and then consider the second part of our construction. The former is described in Section~\ref{sec:tech:flat}. For the latter, we first explain in Section~\ref{sec:tech:sym} how to ``clean up'' the state by removing cross-terms and certain ``asymmetric'' terms. This part is similar in spirit to what is done in previous works (although we present a more general analysis that is based only on flatness and not on specific properties of the input state). Then, we are left with the most technically involved part which is to analyze bound the trace norm of the difference between our state (call it $\rho_{\text{sym}}$) and the state $\rho_{\text{target}}$ above. This is explained in Section~\ref{sec:tech:diff}.

\subsubsection{Flattening}
\label{sec:tech:flat}

Recall that we consider the unitary distribution $H^{\otimes n}U_f$, where $f$ is a random function (i.e.\ a random binary phase followed by Hadamard on all qubits). We say that a vector is $\epsilon$-flat if the square-absolute-value of each of its (standard basis) coefficients is bounded by $\epsilon$. 

Given an input state of the form $\beta = \sum_x \beta_x \ket{x}$, we consider $\gamma_y = \bra{y}H^{\otimes n}U_f\ket{\beta}$, which is the amplitude of the standard basis element $\ket{y}$ in the vector $H^{\otimes n}U_f\ket{\beta}$. This value can be expressed as an exponential sum $\gamma_y = \frac{1}{2^{n/2}}\sum_x f(x) (-1)^{x\cdot y} \beta_x$. We interpret each summand as a random variable with zero mean, since $f$ is a random function to $\{\pm 1\}$. This means that we have a sum of exponentially many independent zero-mean random variables, and furthermore the $\ell_2$ norm of the vector of summands is bounded since $\sum_x |\beta_x|^2 = 1$. This means that the sum is very strongly concentrated around $0$, and indeed using Hoeffding, the square-absolute-value will be at most $\frac{cn}{2^n}$ with all but an exponentially small probability. By applying the union bound, we get that for any a-priori polynomial-size set of input vectors, and for any coefficient $\gamma_y$ of any of these vectors, it holds with all but exponentially small probability that they are all bounded by $\frac{cn}{2^n}$ in square-absolute-value. We note that since we consider complex vectors, the actual analysis separates $\beta$ into its real and imaginary part, and analyze each separately.

From this point and on, we analyze the remainder of our construction under the assumption that the input quantum states are $\frac{cn}{2^n}$-flat.

\subsubsection{Cross-Term Removal and Symmetrization}
\label{sec:tech:sym}

We consider the application of $(U_\pi U_g)^{\otimes st}$ on an input state consisting of $s$ blocks, each of which contains $t$ copies of the same (flat) state, where the vectors in the different blocks are orthogonal. Our goal is to show that, up to normalization, the output state can be expressed as 
\begin{align}\label{eq:tech:symstate}
	\rho_{\text{sym}} = \sum_{z, \sigma} \nu_\sigma \ket{z}\bra{\sigma(z)}~,
\end{align}
where $z$ is summed as in Eq.~\eqref{eq:intro:targetstate}, $\sigma$ ranges over \emph{all} permutations of $[st]$ (not only block-preserving ones), and $\nu_\sigma$ is a term that will be explained below.

As mentioned above, the techniques here are fairly standard in the analysis of pseudorandom states and other objects \cite{ananth23prfs}. However, our analysis relies only on the general notion of flatness and not on the specific expression for the coefficients of the state.

We start by using the flatness of the states. Let $\Pi^\star$ be the projector to the vectors of length $st$ of strings in $\zon$ with unique entries, that is, no entry reoccurs. Then our input state is close to its $\Pi^\star$ projection, since the collision probability in the standard basis of two entries is small due to flatness. Therefore, we may consider an input state whose density matrix supported only over entries $\ket{z}\bra{z'}$, where both $z,z'$ are unique $st$-tuples of elements from $\zon$.

We first apply $U_g^{\otimes st}$, which has the effect of zeroing out the coefficients of $\ket{z}\bra{z'}$ not of the form $\ket{z}\bra{\sigma(z)}$. This is the result of $\Ex_g \left[\prod_i g(z_i) \prod_i g(z'_i)\right]$ being zero for all $z, z'$ which do not have the same entry-histogram since we average over random $g$'s.

Finally, applying $U_\pi^{\otimes st}$ means that the coefficient of $\ket{z}\bra{\sigma(z)}$ becomes independent of $z$, and depends only on $\sigma$. This follows from taking the expectation over $\pi$, which averages the coefficients $\ket{z}\bra{\sigma(z)}$ for all unique entries $z$ (since $\pi$ is a random permutation). We denote the coefficient corresponding to $\sigma$ by $\nu_\sigma$.

\subsubsection{Bounding The Difference}
\label{sec:tech:diff}

The difference between Eq.~\eqref{eq:intro:targetstate} and Eq.~\eqref{eq:tech:symstate} is two-fold. First, the target state only sums over block-preserving permutations, whereas the symmetrized state sums over all permutations. Second, the target state gives the same weight to all terms $\ket{z}\bra{\sigma(z)}$ that it ranges over, whereas the symmetrized state may give different weights to different permutations.

Our crucial observation here, is to notice that if it is possible to go from a permutation $\sigma$ to a permutation $\sigma'$ by performing block-preserving operations, then $\sigma$ and $\sigma'$ have the same coefficient $\nu_\sigma$. This is the case since the input state corresponds to $s$ blocks where each block contains $t$ identical states. Therefore, the state of the system should be invariant under block-preserving permutations, which is manifested in the corresponding terms $\nu_\sigma$ being equal. We may therefore define congruence classes of permutations that differ only by block-preserving operations, and associate the coefficients with the congruence class rather than a specific permutation.

Consider for every congruence class of permutations $\edgep$ the operator $A_\edgep = \sum_{\sigma \in \edgep} \sum_{z} \ket{z}\bra{\sigma(z)}$. Then we can rewrite $\rho_{\text{sym}} =\sum_\edgep \nu_\edgep A_\edgep$. We note that the set of all block-reserving permutations consists of a single congruence class. Therefore, all block-preserving permutations receive the same weight, as required. The remaining goal is to show that the classes that correspond to permutations that are not block-preserving are (jointly) negligible in $\ell_1$ weight.

We let $k$ be the number of ``crossings'' of a congruence class. That is, the number of inputs whose output belongs to a different block. The crossing number, in some sense, represents the amount of deviation from being block-preserving. Denote by $\edgep_k$ the set of congruence classes with $k$ crossings. Note that there is a single congruence class with zero crossings, and it corresponds to the aforementioned set of block-preserving permutations. Our goal therefore is to bound the trace norm of $\sum_{k > 1} \sum_{\edgep \in \edgep_k} \nu_\edgep A_\edgep$.\footnote{We notice that $k=1$ is not possible for reasons of symmetry and therefore it does not appear in the sum, but we could have achieved our result even without this minor optimization.}

The argument here contains three parts:
\begin{enumerate}
	\item We show that for combinatorial reasons, $\norm{A_\edgep}$ grows with $\poly(nst)^k$ (up to a global normalization factor). Essentially, we show that this norm is related to the number of permutations in $\edgep$ (which due to symmetry is the same in all $\edgep$ with the same $k$).
	
	\item We show, again by a combinatorial argument, that the number of congruence classes with the same $k$ also grows as $\poly(nst)^k$.
	
	\item Perhaps the most technically involved part is to show that $\nu_\edgep$ decays, up to a global normalization factor, with $\delta^k$ for a (negligible) factor $\delta = \frac{\poly(nst)}{2^n}$. This is achieved by noticing that $\nu_\edgep$ contains an ``inner product'' term for every crossing edge. This term would ideally correspond to an inner product between two input vectors, and since these are orthogonal we would expect this value to be $0$. However, the inner product is ``disturbed'' because permutations do not allow recurrence, which in turn creates dependence between the would-be inner products. We therefore need to come up with a fairly involved technical argument to show that whereas the value $\nu_\edgep$ is not exactly $0$, each of the would-be inner products contributes a $\delta$ factor, resulting in an exponential decay.
\end{enumerate}

Other factors cancel out perfectly, since they represent the same combinatorial reality, and indeed when putting-together all of the above, we get sum of the form $\sum_{k > 1} \left(\frac{\poly(nst)}{2^n}\right)^k$, which converges to a negligible value as required, bearing in mind that $s,t = \poly(n)$.

\subsection{Other Previous Works}

Alagic, Majenz, and Russell \cite{alagic2020efficient} considered a \emph{stateful} variant of PRU, where the adversary is interacting with a \emph{stateful simulator} whose internal state may change and grow as the experiment proceeds. Even under this relaxed notion, they were only able to construct PRU using a simulator in (quantum) PSPACE. However, in this variant one can hope to achieve \emph{statistical} (or even perfect) security rather than just computational.

\subsection{Future Directions}

Contrary to the constructions in \cite{lu2023prss}, ours is not \emph{scalable} (a notion introduced in \cite{brakerski2020scalable}). That is, we consider adversaries whose running time is polynomial in the input size of the unitary $n$. In contrast, in a scalable construction, one specifies separately the parameters for the adversary's complexity and for the input size. Scalable constructions are usually more involved, and in particularly require more computational depth, than non-scalable ones. It remains an open question to find a scalable version of our construction, or to prove lower bounds in this vein.

Our work shows a fairly strong notion of pseudorandomness for PRU that can be achieved using a straightforward real-valued construction. One may further wonder whether it is possible to get even closer to full-fledged PRU using constructions like ours (or even our construction as-is). For example, the negative result of \cite{haug2023pseudorandom} does not seem to exclude real-valued PRU that are applicable to \emph{tensor-product} inputs. Namely, inputs that are not necessarily orthogonal, but are not entangled with each other (recall that entangled queries stand at the core of the \cite{haug2023pseudorandom} separation). We believe that answering some of these questions may be within reach, but were hurried to report our current progress due to the recent announcement of \cite{metger2024pseudorandom}.

\subsection{Paper Organization}
Section~\ref{section:Preliminaries} introduces notations, defines quantum-secure pseudorandom functions and permutations, and references the notion of almost invariancy under Haar random unitaries. Section~\ref{section:Somewhat-pseudorandom-unitaries} provides the security definition for non-adaptive orthogonal-inputs pseudorandom unitaries, presents the main theorem, describes the construction and reduce it to the information theoretic version. Section~\ref{section:analysis} contains the technical contributions. It starts with the main information theoretic lemma, continues with the analysis of the two steps of the construction separately, and concludes with proving the main lemma and theorem.

\subsection*{Acknowledgments}

We thank Omri Shmueli for multiple contributions to the results presented in this work. We would also like to thank Yanglin Hu and Marco Patrick Tomamichel for pointing out a gap in a proof in the previous version of the manuscript.

\section{Preliminaries} \label{section:Preliminaries}

\subsection{Notation}
We denote $N = 2^n$ and $\Nst = \binom{N}{st}(st)! = N(N-1)\cdots (N-st+1)$. 
We denote the trace norm of $A$ by $\|A\|_1 = \tr (\sqrt{A A^\dagger})$, which is the sum of the singular values.

\paragraph{Vectors, Functions, and Permutations.} Let $s, t \in \bbN$. Throughout our analysis we will consider vectors of bit-strings. Formally, an object of the form $\vec{y} \in (\binset^n)^{st}$, which indicates a $t$-length vector of $n$-bit strings. We will denote $\vec{y} = (y_1, \ldots, y_{st})$, where each coordinate in the vector is a bit-string $y_i \in \binset^n$. 
Let $T$ be some set. We denote by $\uuntt{s}$ the set of all length $st$ vectors of \emph{unique} elements from $\zon$, that is, no entry reoccurs.

We will consider a number of types of operations on such vectors. For any function $f: \binset^n \to D$, where $D$ is some domain, we let $f(\vec{y}) \in D^t$ denote the pointwise application of $f$ on each coordinate in $\vec{y}$ individually. For functions with complex range $\alpha: \binset^n \to \bbC$, we also define multiplicative notation, where we use the notation $\alpha_x = \alpha(x)$, and $\alpha_{\vec{x}} = \prod_{i} \alpha_{x_i}$. Note that the coefficients of quantum states constitute such functions.

For the special case of a function on the coordinates which is a permutation, we will denote it by $\pi: \binset^n \to \binset^n$. We call such a permutation an \emph{inner permutation} since it permutes each element of $\vec{y}$ individually. We also consider an \emph{outer permutation} (or index permutation) $\sigma \in S_{st}$ which permutes $[st] = \{1, \ldots, st\}$. We will abuse the notation and think of $\sigma \in S_{st}$ also as $\sigma \in S_{[s] \times [t]}$ which takes a tuple input $(j, i), j\in [s], i\in [t]$ and outputs the tuple $(j', i')$ s.t.\ $j' = \floor{\sigma(sj + i) / s}, i' = \sigma(sj + i) \mod t$. An outer permutation permutes the indices of the elements of $\vec{y}$, i.e.\ $\vec{z} = \sigma(\vec{y}) \in (\binset^n)^{st}$ is such that $z_{j, i} = y_{\sigma(j, i)}$. We will also consider the subgroup  $S_t^s = S_t \times \cdots \times S_t$ ($s$ times) of outer permutations.

\subsection{Concentration Bounds}
\begin{theorem} [Hoeffding's Inequality]
	Let $Z_x$, $x\in \zon$,  be independent random variables with zero expectation such that  $a_x \le Z_x \le b_x$ with probability $1$. Then for all $\epsilon>0$,
	\begin{align}
		\Pr\left[\sum_{x\in \zon} Z_x \ge \epsilon\right] \le \exp{\left(\frac{-2\epsilon^2}{\sum_i (b_i-a_i)^2}\right)} 
		~.
	\end{align}
\end{theorem}

\subsection{Pseudorandomness}
Zhandry proved that given quantum-secure one-way functions, we can construct quantum-secure pseudorandom functions \cite{zhandry2012construct} and pseudorandom permutations \cite{zhandry2016note}, which we use in our construction.

\begin{definition}[Quantum-Secure Pseudorandom Function] Let $\mathcal{K}$ be a key space. A keyed family of functions $\{F_k: \zon \rightarrow \{\pm 1\}\}_{k\in \mathcal{K}}$ is a quantum-secure pseudorandom function (QPRF) if for any (non-uniform) quantum polynomial-time (QPT) oracle algorithm $\adv$, $F_k$ with a random $k \leftarrow \mathcal{K}$ is indistinguishable from a truly random function $f$ in the sense that
	$$
	\left|
	\Pr_k [\cA^{F_k} (1^n)=1] - \Pr_f   [\cA^{f}(1^n)=1]
	\right| = \negl(n).
	$$
	In addition, $F_k$ is polynomial-time computable on a classical
	computer.
\end{definition}

\begin{definition}[Quantum-Secure Pseudorandom Permutation] Let $\mathcal{K}$ be a key space. A keyed family of permutations $\{P_k: \zon \rightarrow \zon\}_{k\in \cK}$ is a quantum-secure pseudorandom permutation (QPRP) if for any (non-uniform) QPT oracle algorithm $\cA$, $P_k$ with a random $k \leftarrow \cK$ is indistinguishable from a truly random permutation $\pi$ in the sense that
	$$
	\left|
	\Pr_k [\cA^{P_k, P_k^{-1}} (1^n)=1] - \Pr_\pi [\cA^{\pi,\pi^{-1}}(1^n)=1]
	\right| = \negl(n).
	$$
	In addition, $P_k$ is polynomial-time computable on a classical
	computer.
\end{definition}

\subsection{Almost Invariance Under Haar Unitaries}

The following definition, claim and lemma are taken from \cite{ananth2023pseudorandom_isometries}.

\begin{definition} \label{def:almost-inv}
	Let $n, q, \ell \in \bbN$. An $(\ell + n\cdot q)$-qubit state $\rho$ is $\epsilon$-almost invariant under $q$-fold Haar unitary if
	\begin{align}
		\TD \left(
		\rho, \Ex_{U\leftarrow Haar_n} \left[ (I_\ell \otimes U^{\otimes q}) \rho (I_\ell \otimes (U^\dagger)^{\otimes q}) \right]
		\right)
		\le
		\epsilon
	\end{align}
\end{definition}

\begin{tlclaim} \label{claim:channel-almost-inv}
	Let $n, q, \ell \in \bbN$. Suppose $\Phi$ is a quantum channel that is a probabilistic mixture of unitaries on $(\ell + n \cdot q)$ qubits. More precisely, 
	$
	\Phi(\rho) \coloneqq \Ex_{k\leftarrow \mathcal{D}} \left[ (I_\ell \otimes V_k^{\otimes q}) \rho (I_\ell \otimes (V_k^\dagger)^{\otimes q}) \right]
	$
	where $\mathcal{D}$ a distribution over $\binset^*$, and $V_k: \bbC^{2^n} \rightarrow \bbC^{2^n}$ is a unitary for every $k$ in the support of $D$. 
	
	Suppose for a $(\ell + n \cdot q)$-qubit state $\rho$, $\Phi(\rho)$ is $\epsilon$-almost invariant under $q$-fold Haar unitary, where $\epsilon$ is a negligible function, then the following holds:
	
	\begin{align}
		\TD \left(
		\Phi(\rho),
		\Ex_{U\leftarrow Haar_n} \left[ (I_\ell \otimes U^{\otimes q}) \rho (I_\ell \otimes (U^\dagger)^{\otimes q}) \right]
		\right)
		\le 
		\epsilon
	\end{align}
\end{tlclaim}

\begin{lemma} \label{lem:almost-inv-instance}
	Let $n, s, t \in \bbN$, and define
	\begin{align}
		\rho_{\unist} \coloneqq \frac{1}{\Nst} \sum_{\substack{z\in \uuntt{s} \\ \sigma \in S_t^s}} \ket{z} \bra{\sigma(z)}
		~,
	\end{align}
	then for any $\ell$ qubit state $\rho_\ell$, $\rho_\ell \otimes \rho_{\unist}$ is $O(s^2t^2 / 2^n)$-almost invariant under $st$-fold Haar unitary ($I_\ell$ being applied on $\rho_\ell$).
\end{lemma}

This state is close to the output of applying a Haar random unitary on $s$ different orthogonal vectors with $t$ copies of each.

\section{Somewhat Pseudorandom Unitaries} \label{section:Somewhat-pseudorandom-unitaries}

\begin{definition} [Non-Adaptive Orthogonal-Inputs Pseudorandom Unitary] \label{def:naoipru}
	We say that $\{Gen_n\}_{n \in \bbN}$ is a \emph{non-adaptive orthogonal-inputs secure pseudorandom unitary family} if there exists a polynomial $\kappa$ such that:
	\begin{itemize}
		\item For every $k\in \binset^{\kappa(n)}$, $U_k \coloneqq Gen_n(k)$ is a QPT algorithm implementing a unitary operation on $n$ qubits.
		\item Fix $s\coloneqq s(n),t\coloneqq t(n)$ polynomials in $n$. Let $A$ be a set and let $\left\{ \ket{\psi^{(1, a)}} ,\ldots \ket{\psi^{(s, a)}} \right\}$ be orthogonal states and $\rho_a$ be any $\ell$-qubit state for all $a\in A$. Let $p_a$ be probabilities such that $\sum_{a\ in A} p_a = 1$. There exists a sufficiently large $n \in \bbN$, such that for any (non-uniform) QPT distinguisher $\adv$ that makes queries of the form 
		\begin{align}
			\rho_{in} \coloneqq
			\sum_{a \in A}
			p_a 
			\left(
			\rho_a \otimes
			\left(
			\bigotimes_{j \in [s]} (\ketbra{\psi^{(j, a)}})^{\otimes t}
			\right)
			\right)
			~
		\end{align}
		to the $st$-tensor of the unitary $U_k$ it holds that
		\begin{multline}
			\bigg|
			\Pr_{k \leftarrow \binset^\kappa} \left[ \adv \left(
			( I_\ell \otimes U_k^{\otimes st})  \rho_{in} 
			( I_\ell \otimes{U_k^\dagger}^{\otimes st}) \right)
			=1 \right]
			\\-
			\Pr_{U \leftarrow Haar_n} \left[ \adv \left(
			( I_\ell \otimes U^{\otimes st})  \rho_{in} 
			( I_\ell \otimes {U^\dagger}^{\otimes st}) \right)
			=1 \right]
			\bigg|
			\le \negl (n)
		\end{multline}
	\end{itemize}
\end{definition}

\begin{theorem} \label{thm:real-naoipru}
	Assuming the existence of quantum secure one way functions, there exists a family of non-adaptive orthogonal-inputs secure pseudorandom \emph{real} unitary.
\end{theorem}

From definition \ref{def:almost-inv}, claim \ref{claim:channel-almost-inv}, and lemma \ref{lem:almost-inv-instance}, our goal will be to construct a channel $\Phi$ such that its output for copies of orthogonal states looks like  almost invariant under $st$-fold Haar unitary.

\paragraph{The construction}
Let $F,G: \binset^n \to \pmset$ be QPRFs and $P: \binset^n \to \binset^n$ be a QPRP. 
We define the unitary $U_{F,G,P}$ on $n$ qubits as follows:

\begin{align}
	U_{F,G,P} = U_{P} U_G H^{\otimes n} U_{F}
	~,
\end{align}
where $U_F = \sum_{x\in\zon} F(x) \ketbra{x}, U_G= \sum_{x\in\zon} G(x) \ketbra{x}$, and $U_P=\sum_{x\in\zon} \ket{P(x)}\bra{x}$.

\paragraph{Invoking Cryptographic Assumptions}
We move from pseudorandom $F, G, P$ to truly random $f, g, \pi$. The unitary in the random case is denoted by $U_{f, g, \pi}$ and defined similarly.

\begin{tlclaim} \label{claim:pseudo-to-random-ind}
	Let $U_{F, G, P}$ implicitly depend on a key $k$. Assuming the security of $F, G$ and $P$, it holds that
	\begin{multline}
		\bigg|
		\Pr_{k} \left[ \adv \left( 
		\left( I_\ell \otimes {U_{F, G, P}}^{\otimes st} \right)  \rho_{in} 
		\left( I_\ell \otimes {U_{F, G, P}^\dagger}^{\otimes st} \right)
		\right)
		=1 \right]
		\\-
		\Pr_{f, g, \pi} \left[ \adv \left(
		\left( I_\ell \otimes {U_{f, g, \pi}}^{\otimes st} \right)  \rho_{in} 
		\left( I_\ell \otimes {U_{f, g, \pi}^\dagger}^{\otimes st} \right)
		\right)
		=1 \right]
		\bigg| \le \negl(n) ~,
	\end{multline}
	where $k=(k_F, k_G, k_P)$ is the key for the PRFs and PRP.
\end{tlclaim}

\begin{proof}
	We define four hybrids. In the first one we query $U_{F, G, P}$, in the second $U_{f, G, P}$, in the third $U_{f, g, P}$, and in the forth $U_{f, g, \pi}$ (all defined similarly to $U_{F, G, P}$). Each two consecutive hybrids are indistinguishable by the security of function replaced between the two hybrids.
\end{proof}

\section{Analysis} \label{section:analysis}

We now turn to analyze the application of $U_{f, g, \pi}$ information theoretically. The main lemma is:

\begin{lemma} \label{lemma:info-theoretic-bound}
	Let $s, t$ be polynomials in $n$ and let $\{\ket{\psi^{(j)}}\}_{j\in [s]}$ be orthogonal quantum states. Then
	\begin{align}
		\left\| \Ex_{f, g, \pi} \left[
		{U_{f, g, \pi}}^{\otimes st}  
		\left(
		\bigotimes_{j \in [s]} (\ketbra{\psi^{(j)}})^{\otimes t}
		\right)
		{U_{f, g, \pi}^\dagger}^{\otimes st} 
		\right]
		- 
		\rho_{\unist} \right\|_1 \le O(s^6 t^4 n^2 /N) ~.
	\end{align}
	Were the expectation is over sampling random functions $f, g$ and a random permutation $\pi$.
\end{lemma}

\subsection{Achieving Flatness} \label{subsection:Achieving-flatness}

The first two steps in the construction, namely adding a random binary phase with $U_f$ and performing $H^{\otimes n}$, achieve the goal of flattening the state with respect to the standard basis.

\begin{definition}
	A quantum state $\kalpha = \sum_x \alpha_x \ket{x}$, where $\ket{x}$ are the computational basis elements, is $\epsilon$-flat if $\max_x \abs{\alpha_x}^2 \le \epsilon$.
\end{definition}
Note that this is equivalent to the min-entropy of the computational-basis measurement of $\kalpha$ having min-entropy at least $\log(1/\epsilon)$. 

\begin{lemma}
	Let $\ket{\beta} = \sum_{x\in \zon} \beta_x\ket{x}$ be a quantum state, $c>0$, and let $f:\zon \rightarrow \pmset$ be a random function. Then with probability at least $1-2\exp\left(-\left(\frac{c}{4}-\ln(2)\right)n\right)$ the state $H^{\otimes n} U_f \ket{\beta}$ is $c\cdot \frac{n}{2^{n}}$-flat.
\end{lemma}

\begin{proof}
	Denote $\ket{\xi} = H^{\otimes n} U_f \ket{\beta}$. Let $\ket{y}$ be a standard basis element, and look at $\xi_y = \braket{y|\xi}$:
	\begin{align}
		\xi_y=\bra{y} H^{\otimes n} U_f \ket{\beta} &= \frac{1}{2^{n/2}} \sum_{x\in \zon} (-1)^{x\cdot y} \bra{x} U_f \ket{\beta} \\
		&= \frac{1}{2^{n/2}} \sum_{x\in \zon} (-1)^{x\cdot y} f(x) \braket{x|\beta} \\
		&= \frac{1}{2^{n/2}} \sum_{x\in \zon} (-1)^{x\cdot y} f(x) \beta_x
	\end{align}
	
	We analyze the real and imaginary parts $\xi_y$ separately.
	Define the random variables $Z_x$ to be $\Re( 2^{-n/2} (-1)^{x\cdot y} f(x) \beta_x )$. It follows that $|Z_x| \le 2^{-n/2} |\beta_x|$. Using Hoeffding's inequality we get 
	\begin{align}
		\Pr \left[ |\Re(\xi_y)| \ge \sqrt{\epsilon/2} \right] \le
		\exp {\left( \frac{-2\cdot (\epsilon/2)}{\sum_x |2 \cdot 2^{-n/2}\beta_x|^2} \right)} = 
		\exp {\left( \frac{-\epsilon \cdot 2^n}{4 \sum_x |\beta_x|^2} \right)} = 
		\exp(-\epsilon\cdot 2^n/4)
	\end{align}
	
	Where the second to last equality follows from $\ket{\beta}$ being a quantum state and thus a unit vector. We get $\Pr \left[ |\Im(\xi_y)| \ge \sqrt{\epsilon/2} \right] \le \exp(-\epsilon\cdot 2^n/4)$ similarly. Using the bound on both the real part and imaginary part we get:
	
	\begin{align}
		\Pr\left[ |\xi_y| \ge \sqrt{\epsilon} \right] &\le
		\Pr\left[ |\Re(\xi_y)|^2 + |\Im(\xi_y)|^2 \ge \epsilon \right] \\&\le
		\Pr\left[ |\Re(\xi_y)|^2 \ge \epsilon/2 \vee |\Im(\xi_y)|^2 \ge \epsilon/2 \right] \\&\le
		\Pr\left[ |\Re(\xi_y)| \ge \sqrt{\epsilon/2} \right] + \Pr\left[ |\Im(\xi_y)| \ge \sqrt{\epsilon/2}\right] \\&\le 
		2\exp{(-\epsilon\cdot 2^n/4)}
	\end{align}

	Finally, we use the union bound over the $2^n$ entries of $\ket \xi$ to get that with probability at most $2\exp{(-\epsilon\cdot 2^n/4)} \cdot 2^n$ there exists an entry $|\xi_y| \ge \sqrt{\epsilon}$. Taking $\epsilon = c \cdot \frac{n}{2^{n}}$  completes the lemma.
\end{proof}

Using the union bound, an immediate corollary follows:
\begin{corollary} \label{col:poly-flat}
	Let $\{\ket{\beta^{(j)}}\}_{j\in [s]}$ be quantum states, $c>0$, and let $f:\zon \rightarrow \pmset$ be a random function. Then with probability $1-s\cdot 2\exp\left(-\left(\frac{c}{4}-\ln(2)\right)n\right)$ all states $H^{\otimes n} U_f \ket{\beta^{(j)}}$ are $c\cdot \frac{n}{2^{n}}$-flat.
\end{corollary}

\subsection{Getting from Flat States to Random-Looking Ones} \label{subsection:flat-to-random-looking}

We now prove that applying a random binary phase and a random permutation to $t$ copies of $s$ orthogonal flat vectors with is close in trace distance the almost invariant state from lemma \ref{lem:almost-inv-instance}.
We prove the following lemma:

\begin{lemma} \label{lem:state-to-almost-invariant}
	Let $g: \binset^n \to \pmset$ be a random function and $\pi: \binset^n \to \binset^n$ be a random (inner) permutation.
	Let $\{\ket{\alsup{j}}\}_{j \in [s]}$ be a set of $s$ arbitrary orthogonal $\epsilon$-flat vectors in $\bbC^{\binset^n}$.
	Denote:
	\begin{align}
		\rho &\coloneqq \bbE_{g, \pi} \left[
		(U_\pi U_g)^{\otimes st} \left( \otimes_{j \in [s]} \ketbra{\alsup{j}}^{\otimes t} \right) (U_g^\dagger U_\pi^\dagger)^{\otimes st}
		\right]~.
	\end{align}
	Then,
	\begin{align}
		\| \rho - \rho_\unist \|_1 \le O((st)^2\epsilon + N s^6 t^4 \epsilon^2)
		~.
	\end{align}
\end{lemma}

\subsubsection{Focusing on Unique States}

Recall that $g_z = \prod_{j, i} g(z_{j, i})$. We  notice that for all $z, z' \in \binset^{n \cdot st}$ it holds that $\bbE_g [ g_{z} g_{z'}^* ]$ is equal to $1$ if and only if the binary type of $z$ and $z'$ 
(that is, the histogram of the entries modulus 2) are equal. Otherwise, the expectation is equal to $0$. Expressing $\rho$ in the standard basis we get
\begin{align}
	\rho & =
	\Ex_{g, \pi} \left[
	\sum_{z, z'\in \binset^{n\cdot st}}
	g_z g_{z'}
	\prod_{\substack{j \in [s]\\ i\in [t]}}
	\alpha_{z(j,i)}^{(j)} {\alpha^*}^{{(j)}}_{z'(j,i)}
	\ket{\pi(z)}\bra{\pi(z')}
	\right]
	\\ & =
    \Ex_{\pi} \left[
	\sum_{z, z'\in \binset^{n\cdot st}}
	\Ex_{g} [g_z g_{z'}^*]
	\prod_{\substack{j \in [s]\\ i\in [t]}}
	\alpha_{z(j,i)}^{(j)} {\alpha^*}^{{(j)}}_{z'(j,i)}
	\ket{\pi(z)}\bra{\pi(z')}
 	\right]
    \\ & =
        \Ex_{\pi} \left[
	\sum_{z\in \binset^{n\cdot st}}
	\sum_{\substack{z'\in \binset^{n\cdot st} \\ \mathsf{bintype}(z')=\mathsf{bintype}(z)}}
	\prod_{\substack{j \in [s]\\ i\in [t]}}
	\alpha_{z(j,i)}^{(j)} {\alpha^*}^{{(j)}}_{z'(j,i)}
	\ket{\pi(z)}\bra{\pi(z')}
 	\right]
	~,
\end{align}
where $\mathsf{bintype}(z)$ is the binary type of $z$.

Let $\piu \coloneqq \sum_{\vec{z} \in \uuntt{s}} \ketbra{\vec{z}}$ be the uniqueness projector, and let
\begin{align}
	\rho^{\star} \coloneqq \frac{\piu \rho \piu}{\tr[\piu \rho]}
\end{align}
be the unique restrictions of $\rho$. We show that $\rho$ is close to its unique restriction. 

\begin{claim}\label{claim:unique_res}
	It holds that
	\begin{align}
		\| \rho - \rho^{\star} \|_1 \le O \left( (st)^2 \epsilon \right) ~.
	\end{align}
\end{claim}

\begin{claimproof}
Notice that $\piu \rho (I - \piu) = (I -\piu) \rho \piu = 0$, as $\ket{\pi(z)}$ is in the unique restriction if and only if $\bra{\sigma(\pi(z))}$ is in the unique restriction too (which occurs if and only if the binary type/histogram has $st$ entries of $1$). It follows that $\rho = \piu \rho \piu + (I - \piu) \rho (I - \piu)$, and as $\piu \rho \piu$ and $(I - \piu) \rho (I - \piu)$ are positive semi-definite, it is enough to show that
	\begin{align}
		\tr[(I - \piu) \rho (I - \piu)] \le  (st)^2 \cdot \epsilon~.
	\end{align}
	We note that $\piu$ is invariant under $U_g, U_\pi$, therefore
	\begin{align}
		\tr[(I - \piu) \rho (I - \piu)] &=
		\tr\left[(I - \piu) \bbE_{g, \pi} \left[
		(U_\pi U_g)^{\otimes st} \left( \otimes_{j \in [s]} \ketbra{\alsup{j}}^{\otimes t} \right) (U_g^\dagger U_\pi^\dagger)^{\otimes st}\right] \right] 
		\\&=
		\bbE_{g, \pi} \left[ \tr\left[(U_\pi U_g)^{\otimes st} (I - \piu) 
		\left( \otimes_{j \in [s]} \ketbra{\alsup{j}}^{\otimes t} \right) (U_g^\dagger U_\pi^\dagger)^{\otimes st}\right] \right] 
		\\&=
		\bbE_{g, \pi} \left[ \tr\left[ (I - \piu) 
		\left( \otimes_{j \in [s]} \ketbra{\alsup{j}}^{\otimes t} \right) \right] \right] 
		\\&=
		\tr\left[(I - \piu) \left(\otimes_{j \in [s]} \ketbra{\alsup{j}}^{\otimes t}\right)\right]~,
	\end{align}
	which is exactly the probability of measuring $\otimes_{j \in [s]} \ketbra{\alsup{j}}^{\otimes t}$ in the computational basis, and obtaining an $(st)$-tuple that contains a repetition (i.e.\ an element in $\zon$ that appears more than once). Due to $\epsilon$-flatness, the probability for this is bounded by $(st)^2 \cdot \epsilon$. 
\end{claimproof}

Recall that $\uuntt{s}$ is the set of all $st$ length vectors with unique entries from $\zon$. From the definitions of $U_g, U_\pi$ and claim \ref{claim:unique_res}, for $c_1 = \tr[\piu \rho_{g, \pi}] \ge \frac{1}{1-  \epsilon (st)^2}$ we get

\begin{align}
	\rho^{\star} & =
	c_1 \cdot     	\Ex_{\pi} \left[
	\sum_{z\in \uuntt{s}}
	\sum_{\sigma \in S_{st}}
	\prod_{\substack{j \in [s]\\ i\in [t]}}
	\alpha_{z(j,i)}^{(j)} {\alpha^*}^{{(j)}}_{\sigma(z)(j,i)}
	\ket{\pi(z)}\bra{\sigma(\pi(z))}
	\right]
	~.
\end{align}
Notice the sum over $z'$ changed to sum over $\sigma \in S_{st}$ (an outer permutation which permutes the positions of the entries) as for $z$ with unique entries it holds that $bintype(z)=bintype(z')$ if and only if there exists $\sigma \in S_{st}$ s.t. $z'=\sigma(z)$.

For all $(j,i) \in [s] \times [t]$ we define $\alsup{j,i} = \alsup{j}$ (since we implicitly think of the index $(j,i)$ as pointing to the $j$'th qubit group which consists of a $t$-tensor of $\ketbra{\alsup{j}}$).
Changing the order of summation by $\pi^{-1}$. We get
\begin{align}
	\rho^{\star} & =
	c_1 \cdot \sum_{\sigma \in S_{st}} \sum_{z \in \uuntt{s}}
	\underbrace{
		\bbE_{\pi^{-1}} \left[
		\prod_{\substack{j \in [s]\\ i\in [t]}}
		\alpha_{\pi^{-1}(z)(j,i)}^{(j,i)} {\alpha^*}^{{(j,i)}}_{\sigma(\pi^{-1}(z))(j,i)}
		\right]
	}_{\text{denote $\nu_{\sigma,z}$}} 	
	\ket{z}\bra{\sigma(z)}	
	~.
\end{align}

As $\pi^{-1}$ is also a random permutation, we get that $\nu_{\sigma, z}$ is independent of $z$, i.e.\ $\nu_{\sigma, z} = \nu_{\sigma}$ for all $z\in \uuntt{s}$. Making a change of variables $x = \pi^{-1}(z)$,
\begin{align}
	\nu_{\sigma} =
	\bbE_{x\in \uuntt{s}} \left[
	\prod_{\substack{j \in [s]\\ i\in [t]}}
	\alpha_{x(j,i)}^{(j,i)} {\alpha^*}^{{(j,i)}}_{\sigma(x)(j,i)}
	\right]
\end{align}
and
\begin{align}
	\rho^{\star} & = c_1 \cdot\sum_{\sigma \in S_{st}} \nu_{\sigma} \sum_{z\in \uuntt{s}}
	\ket{z}\bra{\sigma(z)}
	~.
\end{align}

\subsubsection{Using Orthogonality to Reach Closeness to an Almost Invariant State}

We consider the operator
\begin{align}
	A = \sum_{\sigma  \in S_{st}} \nu_{\sigma} \sum_{z\in \uuntt{s}}
	\ket{z}\bra{\sigma(z)}~,
\end{align}
and show that it is close in trace norm to to the same operator summing only over $\sigma \in S_t^s$. For that, we define the following.

\begin{definition}
	For any permutation $\sigma \in S_{st}$, we consider the associated directed graph $G_{\sigma}$, whose vertex set is $[s] \times [t]$, and there is an edge $(j, i) \to (j', i')$ if and only if $\sigma((j,i)) = (j',i')$. 
	For all $j \in [s]$, we define the $j$-th vertex-block as the set $\{ (j,i) \}_{i \in [t]}$. We sometimes completely associate $\sigma$ with $G_\sigma$.
	
	We associate each vertex with its outgoing edge. For any vertex $v=(j,i) \in G_\sigma$ with the outgoing edge $(j',i') =\sigma((i, j))$, we denote $j_v=j$, $j'_v = j'$, namely the block-source and block-destination of $v$ in the graph.
	We say that $v$ is a \emph{crossing vertex} if $j_v \ne j_v'$, and otherwise $v$ is \emph{non-crossing}. We denote the set of crossing vertices by $\csv_\sigma$, and will often omit the subscript  when $\sigma$ is clear from the context. Likewise, we denote the set of non-crossing vertices by $\overline{\csv_\sigma}$.
	
	The \emph{block edge pattern} of $\sigma$ is the vector
	$\edgep_\sigma \in \bbN^{[s]\times [s]}$, where $\edgep_\sigma[j,j']$ is the number of crossing vertices from block $j$ to block $j'$. 
	We say that two permutations are congruent (with respect to $S_t^s$) if they have the same block edge pattern. It follows that $\sigma, \sigma'$ are congruent, denoted $\sigma \cong \sigma'$ if and only if there exist $\sigma_1, \sigma_2 \in S_t^s$ s.t.\ $\sigma' = \sigma_1 \sigma \sigma_2$. We overload the notation and use $\edgep$ also to denote the congruence class corresponding to this pattern.
	
	The number of crossing and non-crossing vertices is thus $\abs{\csv}$, and $\abs{\ncsv}$ respectively (so, $\abs{\csv} + \abs{\ncsv} = st$). We note that $\abs{\csv}$, $\abs{\ncsv}$ only depend on the congruence class $\edgep$. 
\end{definition}

Under the above definition, and denoting $x(v) = x(j,i)$, we have that

\begin{align}
	\nu_\sigma &= 
	\bbE_{x\in \uuntt{s}} \left[
	\prod_{\substack{j \in [s]\\ i\in [t]}}
	\alpha_{x(j,i)}^{(j,i)} {\alpha^*}^{{(j,i)}}_{\sigma(x)(j,i)}
	\right]
	=
	\bbE_{x\in \uuntt{s}} \left[
	\prod_{\substack{j \in [s]\\ i\in [t]}}
	\alpha_{x(j,i)}^{(j,i)}
	\prod_{\substack{j \in [s]\\ i\in [t]}}
	{\alpha^*}^{{(j,i)}}_{x(\sigma^{-1}(j,i))}
	\right]
	\\
	& =
	\bbE_{x\in \uuntt{s}} \left[
	\prod_{\substack{j \in [s]\\ i\in [t]}}
	\alpha_{x(j,i)}^{(j,i)}
	\prod_{\substack{j \in [s]\\ i\in [t]}}
	{\alpha^*}^{{\sigma((j,i))}}_{x(j,i)}
	\right]
	=
	\bbE_{x\in \uuntt{s}} \left[
	\prod_{\substack{j \in [s]\\ i\in [t]}}
	\alpha_{x(v)}^{(j_v)}
	{\alpha^*}^{(j_v')}_{x(v)}
	\right]
	\\&=\label{eq:sigmawithv}
	\bbE_{x\in \uuntt{s}} \left[
	\prod_{v \in \ncsv} \abs{\alsup{j_v}_{\vxsup{v}}}^2 \prod_{v \in \csv} \alsup{j_v}_{\vxsup{v}} \alssup{j'_v}_{\vxsup{v}} 
	\right]
	~.
\end{align}

\begin{proposition}
	If $\sigma \cong \sigma'$ then $\nu_\sigma = \nu_{\sigma'}$.
\end{proposition}
\begin{proof}
	Consider a permutation $\sigma$. Let us break the operand in the expectation in Eq.~\eqref{eq:sigmawithv} into blocks. Namely, for a fixed $j$ consider
	\begin{align}
		\prod_{\substack{v \in \ncsv \\ j_v = j}} \abs{\alsup{j_v}_{\vxsup{v}}}^2 \prod_{\substack{v \in \csv \\ j_v=j}} \alsup{j_v}_{\vxsup{v}} \alssup{j'_v}_{\vxsup{v}}
		~.
	\end{align}
	We notice that the expression above only depends on the number and block identities of the neighbors of the elements in the $j$'th block. Multiplying over all blocks we get $\nu_\sigma$ which remains invariant under (outer) permutations in $S_t^s$.
\end{proof}

We can therefore denote $\nu_{\edgep}$ which is the value corresponding to $\nu_\sigma$ for all $\sigma \in \edgep$. Define
\begin{align}
	A_{\edgep} & \coloneqq \sum_{z\in \uuntt{s}} \sum_{\sigma \in \edgep}  \ket{{z}}\bra{\sigma({z})}
	= \sum_{z\in \uuntt{s}} \ket{{z}} \sum_{\sigma \in \edgep}  \bra{\sigma({z})}~.
\end{align}

\begin{corollary}
	It holds that 
	\begin{align}
		A & = \sum_{\edgep} \nu_{\edgep} \sum_{z\in \uuntt{s}} \sum_{\sigma \in \edgep}  \ket{{z}}\bra{\sigma({z})}
		=   \sum_{\edgep} \nu_{\edgep} A_{\edgep}~,
	\end{align}
\end{corollary}

We separate the analysis to bounding the norm of $A_\edgep$ according to the crossing number of $\edgep$, counting the number of congruence classes with a certain crossing number, and bounding $\nu_\edgep$ according to the crossing number of $\edgep$.

\begin{lemma}\label{lem:boundaep}
	Let $\edgep$ be with $\cnum=k$. It holds that
	\begin{align}
		\norm{A_\edgep}_1 \le  \Nst \cdot t^k
	\end{align}
\end{lemma}
\begin{proof}
	Partition the space $z\in \uuntt{s}$ into parts that are invariant under ``block-permutations'', i.e.\ under $S_t^s$. 
	By definition, each such set contains $(t!)^s$ different ${z}$ values (recall that all $z(j,i)$ are unique), and the number of partitions is $\frac{\Nst}{(t!)^s}$.
	
	For each partition $P$, define 
	\begin{align}
		A_{\edgep,P} = \sum_{{z}\in P}  \ket{{z}} \sum_{\sigma \in \edgep}\bra{\sigma({z})} ~.
	\end{align}
	For all $z\in P$, the vector $\sum_{\sigma \in \edgep}\bra{\sigma({z})}$ is the same, since the elements of $P$ all differ by a $\tilde{\sigma} \in S_t^s$ permutation, and $\edgep = \edgep \tilde{\sigma}$. Thus, $A_{\edgep,P}$ is a rank-$1$ matrix, and the norm $\norm{A_{\edgep,P}}_1$ is the product of the Euclidean norm of the two vectors. In our case,
	\begin{align}
		\norm{A_{\edgep,P}}_1 &= \sqrt{\abs{P}} \cdot \sqrt{\abs{\edgep}}\\
		& = (t!)^{s/2} \cdot \sqrt{\abs{\edgep}}~.
	\end{align}
	Therefore, by the triangle inequality we have that
	\begin{align} \label{eq:bound-ap}
		\norm{A_{\edgep}}_1 \le \frac{\Nst}{(t!)^s} \cdot (t!)^{s/2} \cdot \sqrt{\abs{\edgep}}~,
	\end{align}

	Next, we bound the cardinality of $\edgep$ when interpreted as a congruence class (namely, the number of permutations that have edge pattern $\edgep$):
	\begin{align} \label{eq:bound-pk}
		\abs{\edgep} \le 
		(t!)^s \cdot t^{2k}~.
	\end{align}
	We show this by over-counting the set of graphs with a given edge pattern and $\cnum=k$. We first consider all of the crossing edges, of which there are $k$ by definition. For each such edge, the edge pattern already specifies its source and destination blocks, so we need to choose its specific source and origin vertices within the blocks. There are at most $t^2$ options for each edge, and thus at most $t^{2k}$ options in general. Then, for all of the other edges, it just remains to go over each of the $s$ blocks of vertices in the graph, and organize the internal edges in the block. We note that any such arrangement can be completed into a permutation in $S_t$, by orienting the outgoing and incoming edges of the block towards each other arbitrarily. Therefore, the number of arrangements in each block is at most $\abs{S_t} = t!$. It follows that across all blocks, the total number of internal arrangements in bounded by $(t!)^s$. The lemma follows from equations \ref{eq:bound-ap} and \ref{eq:bound-pk}.
	
\end{proof}

\begin{lemma}\label{lem:boundpatterns}
	Denote
	\begin{align}
		\edgep_k = \left\{ \edgep : \cnum=k \right\} 
	\end{align}
	the number of congruence classes with crossing number $k$. Then $\abs{\edgep_k} \le s^{2k}$.
\end{lemma}
\begin{proof}
	We over-count the elements in $\edgep_k$. Each edge should be assigned to one of $\binom{s}{2} \le s^2$ pairs of origin and destination blocks. Therefore, the total number of edge patterns is at most $s^{2k}$.
\end{proof}

We now bound 

\begin{lemma}\label{lem:boundnuedge}
	Let $\edgep$ be with crossing numeber $\cnum$. It holds that
	\begin{align}
		\abs{\nu_\edgep} \le (\Nst)^{-1} ((st)^2 \epsilon^2 N)^{\cnum/2}~.
	\end{align}
\end{lemma}
\begin{proof}
	Consider some $\sigma \in \edgep$, and recall that
	\begin{align}
		\nu_\sigma  &= 
		\bbE_{x\in \uuntt{s}} \left[
		\prod_{v \in \ncsv} \abs{\alsup{j_v}_{\vxsup{v}}}^2 \prod_{v \in        \csv} \alsup{j_v}_{\vxsup{v}} \alssup{j'_v}_{\vxsup{v}} 
		\right] =
		(\Nst)^{-1} \sum_{x\in \uuntt{s}} 
		\left[
		\prod_{v \in \ncsv} \abs{\alsup{j_v}_{\vxsup{v}}}^2 \prod_{v \in \csv} \alsup{j_v}_{\vxsup{v}} \alssup{j'_v}_{\vxsup{v}}
		\right]
		~.
	\end{align}
	It is possible to sum over the $x$ elements as follows. Go over all $v$ at arbitrary order, and for each $v$, let $x(v)$ run over all elements in $\binset^n$ that were not selected in the previous $v$'s, for each such value of $x(v)$ continue to pick value for the next $v$ in the order.

	In order to analyze this expression, we consider a more general expression as follows:
	\begin{align}\label{eq:tauformgen}
		\tau = \sum_{y_1,\ldots, y_m} \prod_{i\in[m]} \abs{\alsup{j_i}_{y_i}}^2 
		\delta(\alpha_{y_1},\ldots, \alpha_{y_m}) 
		\sum_{z_1} M_1(\alpha_{z_1}) \cdots \sum_{z_\ell} M_\ell(\alpha_{z_\ell})~,
	\end{align}	
	where the indices $y_i$ and $z_i$ run over all of $\binset^n$ except 
	for the preceding values of the indices. That is, the $y_1,\ldots, y_m$ values are first chosen to be distinct, and then each $z_i$ is chosen in order to be distinct from $y_1,\ldots, y_m$ and all preceding $z_i$.
	
	The function $\delta(\cdot)$ can be an arbitrary polynomial, and the functions $M_i$ are monomials, where $\delta$ and $M_i$ can act on the set of their operands and their complex conjugates. We only consider setting where the total degree of $M_i$ is even. We note that we can always reorder the $z_i$'s without effecting the total value of the expression (maintaining the convention that ``later'' $z_i$'s take all values except those of $y_i$'s and ``previous'' $z_i$'s and).
	
	We let $d_i$ denote the total degree of $M_i$. We say that an index $z_i$ is \emph{loaded} if $d_i \ge 4$, otherwise we say that it is \emph{free} (by our convention this means that $d_i=2$). As a convention, we always order the summation so that the loaded indices are enumerated on before the free ones.
	We let $\ell'$ denote the number of loaded indices. Furthermore, in our setting, any free term $i$ is going to be of the form $M_i(\alpha_{z_i}) = \alsup{j}_{z_i} \alssup{j'}_{z_i}$, or its complex conjugate, for some $j \neq j'$.
	
	We define the \emph{magnitude} of $\tau$ as follows, letting $\delta_0$ be the maximal value of $\delta$ over all possible inputs that it can take:
	\begin{align}\label{eq:magtaubound}
		\magn(\tau) &= \delta_0 \prod_{i \in [\ell']} (\epsilon^{d_i/2}N)\\
		&= \delta_0 N^{\ell'} \epsilon^{(\sum_i d_i)/2}~.
	\end{align}
	We let $d = \sum_{i\in[\ell']} d_i$ denote the total degree of all loaded elements.

	We now recall that $\ket{\alsup{j}}, \ket{\alsup{j'}}$ are orthogonal for $j \neq j'$, and therefore $\sum_{z \in \binset^n} \alsup{j}_{z} \alssup{j'}_{z} = 0$. It therefore follows that if $\tau$ has any free indices, i.e.\ by our convention if its $\ell$ index is free, then we have that (up to complex conjugation)
	\begin{align}
		\sum_{z_\ell} M_{\ell}(\alpha_{z_\ell})
		=
		\sum_{z_\ell} \alsup{j}_{z_\ell}\alssup{j'}_{z_\ell}
		=
		0 - \sum_{i < \ell} \alsup{j}_{z_i}\alssup{j'}_{z_i} + \delta_\ell(\alpha_{y_1},\ldots \alpha_{y_m}) ~.
	\end{align} 
	where $\abs{\delta_\ell} \le \epsilon \cdot (st)$ from $\epsilon$ flatness.
	
	Each term of the form $\alsup{j}_{z_i}\alssup{j'}_{z_i}$ now ``joins'' $M_i$ and so it either creates a new loaded term, if $z_i$ was not previously loaded,  decreasing the value of $\magn{}$ by $\epsilon^2 N$, or adds $2$ to the degree of a pre-existing loaded term if $z_i$ was previously loaded, decreasing the value of $\magn{}$ by $\epsilon$. Furthermore, multiplying by $\delta_\ell$ would decrease the value of the ``global'' delta by a factor of $\epsilon \cdot (st)$. Therefore, we can write $\tau$ as a sum of $\ell$ terms, each of which conforms with the general form of Eq.~\eqref{eq:tauformgen}, but with only $\ell-1$ indices (rather than $\ell$), namely:
	\begin{align}
		\tau = \sum_{i \in [\ell]} \tau_i~,
	\end{align}
	and it holds that
	\begin{align}\label{eq:magtauigen}
		\magn(\tau_i) \le \max\{\epsilon^2 N, \epsilon, \epsilon \cdot (st)\} \magn(\tau) \le (\epsilon^2 N (st)) \cdot \magn(\tau)~.
	\end{align}
	
	We can now prove the following inductive claim:
	\begin{claim}\label{claim:boundtgen}
		Let $\tau$ be with parameters $m, \ell, \ell'$ as above, and assume $(\epsilon^2 N (st) \ell) < 1$ then it holds that
		\begin{align}
			\abs{\tau} \le (\epsilon^2 N (st) \ell)^{\frac{\ell-\ell'}{2}} \cdot \magn(\tau)~.
		\end{align}
	\end{claim}
	\begin{claimproof}
		We prove this inductively over the value of $\ell-\ell'$.
		For the base case, consider the setting where $\ell' = \ell$. Notice that $\sum_{y_1,\ldots,y_{m}} \prod_{i\in[m]} \abs{\alsup{j_i}_{y_i}}^2  \le 1$. Therefore, $\tau$ is a (sub) convex combination of elements of the form $\delta(\alpha_{y_1}, \ldots, \alpha_{y_m}) \sum_{z_1} M_1(\alpha_{z_1}) \cdots \sum_{z_\ell} M_\ell(\alpha_{z_\ell})$, that are each bounded in absolute value by  $\delta_0 N^{\ell'} \epsilon^{d/2}$, which we show below. It follows that if $\tau$ has no free terms, then $\abs{\tau} \le \magn(\tau)$, where $\magn(\tau)$ is given by Eq.~\eqref{eq:magtaubound}.
		
		Indeed, each such element is a product of $\delta$, times a product of $\ell'$ loaded sums. The total number of summands over all the sums is at most $N^{\ell'}$ (as $l' = l$). Each element in the sum is a product of $d_i$ elements from $\alpha$. By flatness, each element of $\alpha$ has absolute value at most $\sqrt{\epsilon}$, and therefore each element in the sum has absolute value at most $\epsilon^{d_i/2}$. We get a value that is bounded by $\delta_0 N^{\ell'} \epsilon^{d/2}$ as required.
		
		Now consider the case where $\ell > \ell'$. In this case, we can write $\tau = \sum_{i \in [\ell]} \tau_i$ as above. We notice that for each $\tau_i$, we have $\ell_i = \ell-1$, and $\ell'_i \le \ell'+1$, so $\ell-\ell'$ shrinks by at most $2$. Therefore we get the bound 
		\begin{align}
			\abs{\tau} &\le \sum_{i \in [\ell]} \abs{\tau_i} & 
			\\  & \le \sum_i (\epsilon^2 N (st) \ell_i)^{\frac{\ell_i-\ell'_i}{2}} \cdot \magn(\tau_i) & \text{\tiny (induction)}
			\\  & \le \sum_i (\epsilon^2 N (st) \ell)^{\frac{\ell-\ell'}{2}-1} \cdot \magn(\tau_i) & \text{\tiny ($l_i \le l$, ${\ell_i-\ell'_i} \ge {\ell-\ell'}-2$)}
			\\  & \le \ell \cdot (\epsilon^2 N (st) \ell)^{\frac{\ell-\ell'}{2}-1} \cdot (\epsilon^2 N (st)) \cdot \magn(\tau) & \text{\tiny (Eq.~\eqref{eq:magtauigen})}
			\\  & \le  (\epsilon^2 N (st) \ell)^{\frac{\ell-\ell'}{2}} \cdot \magn(\tau) & 
		\end{align}
		and the claim thus follows.	
	\end{claimproof}

	Now, let us go back to our expression for $\nu_\sigma$. We can write it as $\nu_\sigma = (\Nst)^{-1}\tau$, where $\tau$ has the form as above, with $\delta=1$, $\ell = |\csv| \le st$, and $\ell'=0$, thus $\magn(\tau)=1$. Claim~\ref{claim:boundtgen} therefore guarantees that
	\begin{align}
		\abs{\tau} \le (\epsilon^2 N (st)^2)^{\cnum/2}~,
	\end{align}
	and the bound for $\nu_\sigma$ thus follows.
\end{proof}

\begin{corollary}
	Let $\gamma = \epsilon^2 N (st)^2$, and assume ${t^2 s^4\gamma} < 1/4$, then it holds that
	\begin{align}
		\sum_{\edgep : \cnum  > 0} \norm{\nu_{\edgep} A_{\edgep}}_1 \le 2 t^2 s^4 \gamma = 2t^4 s^6 \epsilon^2 N
	\end{align}
\end{corollary}
\begin{proof}
	We derive the bound in the following equation. We note that $\cnum$ cannot be equal to $1$ since for every outgoing edge from a block there needs to be an incoming edge. We also recall the notation of $\edgep_k = \{\edgep : \cnum=k \}$, introduced in Lemma~\ref{lem:boundpatterns} (namely, $\edgep_k$ is a set whose elements are congruence classes $\edgep$).
	\begin{align}
		\sum_{\edgep : \cnum  > 0} \norm{\nu_{\edgep} A_{\edgep}}_1 &= \sum_{k=2}^{st} \sum_{\edgep \in \edgep_k} \abs{\nu_{\edgep}} \norm{A_{\edgep}}_1 &
		\\
		&\le \sum_{k=2}^{st} \sum_{\edgep \in \edgep_k} \gamma^{k/2} t^k & \text{\tiny Lemmas~\ref{lem:boundnuedge} and~\ref{lem:boundaep}} 
		\\
		&\le \sum_{k=2}^{st} \gamma^{k/2} t^k s^{2k}  & \text{\tiny Lemma~\ref{lem:boundpatterns}}
		\\&= \sum_{k=2}^{st} (t s^2 \sqrt{\gamma})^{k}  & 
		\\& \le \frac{(t s^2 \sqrt{\gamma})^2}{1- t s^2 \sqrt{\gamma}}
		\\& \le 2 t^2 s^4 \gamma
	\end{align}
	and the lemma follows.
\end{proof}

We can now prove lemma \ref{lem:state-to-almost-invariant}
\begin{proof} [Proof of lemma \ref{lem:state-to-almost-invariant}]
	Using the triangle inequality,
	\begin{align}
		\|  \rho - \rho_\unist \|_1 
		\le 
		\|  \rho - \rho^{\star} \|_1 
		+
		\left\| \rho^{\star} - c_1 \cdot \sum_{\edgep:|\csv| = 0}  \nu_\edgep A_\edgep \right\|_1 
		+ \left\| c_1 \cdot \sum_{\edgep:|\csv| = 0}  \nu_\edgep A_\edgep - \rho_\unist \right\|_1
		~.
	\end{align}
	We bound each term separately.
	\begin{align}
		\| \rho^\star - c_1 \cdot \sum_{\edgep:|\csv| = 0}  \nu_\edgep A_\edgep \|_1
		&= 
		\| c_1 \cdot \sum_{\edgep:|\csv| > 0}  \nu_\edgep A_\edgep \|_1
		\\ &\le
		c_1 \cdot \sum_{\edgep:|\csv| > 0} \| \nu_\edgep A_\edgep \|_1
		\\& \le c_1 \cdot 2 t^2 s^4 \gamma
		\\ &=
		c_1 \cdot 2 \epsilon^2 N s^6 t^4
	\end{align}
	
	We note that there is one congruence class with zero crossing, which contains the permutations $\sigma \in S_t^s$. Denote this congruence class by $\edgep_0$, so $\sum_{\edgep:|\csv| = 0}  A_\edgep = A_{\edgep_0}$. Recall that 
	\begin{align}
		\rho_\unist = \frac{1}{\Nst} \sum_{\substack{z\in \uuntt{s} \\ \sigma \in S_t^s}} \ket{z} \bra{\sigma(z)}
		=
		\frac{1}{\Nst} A_{\edgep_0}
	\end{align}
	As $\tr(\rho^{\star}) = 1$, we have
	\begin{align}
		\left| 1 -  c_1 \cdot  \nu_{\edgep_0} \Nst \right|
		=
		\left| 1 - \tr \left( c_1 \cdot  \nu_{\edgep_0} A_{\edgep_0} \right) \right| \le c_1 \cdot 2 \epsilon^2 N s^6 t^4
		~.
	\end{align}
	It follows that 
	\begin{align}
		\frac{1-c_1 \cdot 2 \epsilon^2 N s^6 t^4}{c_1}
		\frac{1}{\Nst}
		\le
		\nu_{\edgep_0} 
		\le
		\frac{1+ c_1 \cdot 2 \epsilon^2 N s^6 t^4}{c_1}
		\frac{1}{\Nst}
	\end{align}
	and so,
	\begin{align}
		\left\| c_1 \cdot \sum_{\edgep:|\csv| = 0}  \nu_\edgep A_\edgep - \rho_\unist \right\|_1 
		=
		\left\| c_1 \cdot \nu_{\edgep_0} A_{\edgep_0} - \frac{1}{\Nst} A_{\edgep_0} \right\|_1 
		\le
		c_1 \cdot 2 \epsilon^2 N s^6 t^4
	\end{align}
	and it follows that
	\begin{align}
		\|  \rho^{\star} - \rho_\unist \|_1 \le O(\epsilon^2 N s^6 t^4)
	\end{align}
\end{proof}

\subsection{Proving the main lemma and theorem}
We combine the results of sections \ref{subsection:Achieving-flatness} and \ref{subsection:flat-to-random-looking} to prove the main lemma.

\begin{proof} [Proof of lemma \ref{lemma:info-theoretic-bound}]
	Let $s, t$ be polynomials in $n$ and let $\{\ket{\psi^{(j)}}\}_{j\in [s]}$ be orthogonal quantum states.
	From corollary \ref{col:poly-flat}, choosing $c=8$ we get that $\forall_{j \in [s]} H^{\otimes n}U_f\ket{\psi^{(j)}}$ is $\frac{8n}{N}$-flat with probability at least $1-s\cdot 2\exp\left(-\left(\frac{8}{4}-\ln(2)\right)n\right) \ge 1-2^{-n} = 1-\frac{1}{N}$.
	
	Assuming flatness of these states, we can now use lemma \ref{lem:state-to-almost-invariant} and get that:
	\begin{align}
		\left\|
		\Ex_{f, g, \pi}
		\left[ {U_{f, g, \pi}}^{\otimes st}  
		\left(
		\bigotimes_{j \in [s]} (\ketbra{\psi^{(j)}})^{\otimes t}
		\right)
		{U_{f, g, \pi}^\dagger}^{\otimes st} 
		\right]
		-
		\rho_\unist
		\right\|_1 &\le 
		O\left(\frac{(st)^2 8n}{N} + N s^6 t^4 \left(\frac{9n}{N}\right)^2 + \frac{1}{N} \right)
		\\&
		=O\left(\frac{s^6 t^4 n^2}{N}\right)
		~.
	\end{align}
\end{proof}

We conclude with a proof for theorem \ref{thm:real-naoipru}.

\begin{proof}  [Proof of theorem \ref{thm:real-naoipru}]
	Taking $Gen_n(k) = U_k$ to be $U_{F, G, P}$ (where $k$ is split into keys for $F,G$ and $P$), we get that it is indeed a QPT algorithm on $n$ qubits. We now prove the security requirement.
	
	Recall that $\rho_{in}$ is promised to be of the form
	\begin{align}
		\rho_{in} =
		\sum_{a \in A}
		p_a 
		\left(
		\rho_a \otimes
		\left(
		\bigotimes_{j \in [s]} (\ketbra{\psi^{(j, a)}})^{\otimes t}
		\right)
		\right)
	\end{align}
	for orthogonal sets of states $\left\{ \ket{\psi^{(1, a)}}, \ldots , \ket{\psi^{(s, a)}} \right\}$. Define the channel $\Phi$ to be 
 \begin{align}
      \Phi(\rho) = \Ex_{f, g, \pi}
	\left[ \left(I_\ell \otimes {U_{f, g, \pi}}^{\otimes st} \right)
	\rho 
	\left(I_\ell \otimes {U_{f, g, \pi}^\dagger}^{\otimes st} \right)
	\right]
 ~.
  \end{align}
	Performing $\Phi$ on $\rho_{in}$ results in the state
	\begin{align}
		\Phi(\rho_{in}) = 
		\sum_{a \in A}
		p_a 
		\left(
		\rho_a \otimes
		\Ex_{f, g, \pi}
		\left[
		{U_{f, g, \pi}}^{\otimes st}
		\left(
		\bigotimes_{j \in [s]} (\ketbra{\psi^{(j, a)}})^{\otimes t}
		\right)
		{U_{f, g, \pi}^\dagger}^{\otimes st}
		\right]
		\right) ~.
	\end{align}
	By lemma \ref{lemma:info-theoretic-bound} and the fact that $\sum_{a\in A} p_a = 1$ we get that
	\begin{align}
		\left\| \Phi(\rho_{in}) - \sum_{a \in A} p_a (\rho_a \otimes \rho_\unist) \right\|_1 \le O\left(\frac{s^6 t^4 n^2}{N}\right) ~.
	\end{align}
	
	Together with lemma \ref{lem:almost-inv-instance}, we get by the triangle inequality that $\Phi(\rho_{in})$ is an $O(s^6 t^4 n^2 /N + s^2t^2/N) = O(s^6 t^4 n^2 /N)$ almost invariant state.
	Using claim \ref{claim:channel-almost-inv} we get 
	\begin{align} \label{eq:almost-inv-result}
		\TD \left(
		\Phi(\rho_{in}),
		\Ex_{U\leftarrow Haar_{n}} \left[ ( I_\ell \otimes U^{\otimes st}) \rho_{in} (I_\ell \otimes (U^\dagger)^{\otimes st}) \right]
		\right)
		\le 
		O(s^6 t^4 n^2 /N)~.
	\end{align}
	By claim \ref{claim:pseudo-to-random-ind} we have 
	\begin{multline} \label{eq:FGP-fgp-negl}
		\bigg|
		\Pr_{k} \left[ \adv \left( 
		\left( I_\ell \otimes {U_{F, G, P}}^{\otimes st} \right)  \rho_{in} 
		\left( I_\ell \otimes {U_{F, G, P}^\dagger}^{\otimes st} \right)
		\right)
		=1 \right]
		\\-
		\Pr_{f, g, \pi} \left[ \adv \left(
		\left( I_\ell \otimes {U_{f, g, \pi}}^{\otimes st} \right)  \rho_{in} 
		\left( I_\ell \otimes {U_{f, g, \pi}^\dagger}^{\otimes st} \right)
		\right)
		=1 \right]
		\bigg| \le \negl(n) ~.
	\end{multline}
	We finish by combining equations \ref{eq:FGP-fgp-negl} and \ref{eq:almost-inv-result} to get
	\begin{multline}
		\bigg|
		\Pr_{k} \left[ \adv \left( 
		\left( I_\ell \otimes {U_{F, G, P}}^{\otimes st}  \right)
		\rho_{in} 
		\left( I_\ell \otimes {U_{F, G, P}^\dagger}^{\otimes st} 
		\right)
		\right)
		=1 \right]
		\\-
		\Pr_{U\leftarrow Haar_{n}} \left[ \adv \left( (I_\ell \otimes {U}^{\otimes st})  \rho_{in} 
		(I_\ell \otimes {{U}^\dagger}^{\otimes st}) \right)
		=1 \right]
		\bigg|
		\le \negl(n) + O(s^6 t^4 n^2 /N) = \negl(n) ~,
	\end{multline}
	as needed to satisfy the security definition \ref{def:naoipru}.
	
\end{proof}

\bibliographystyle{alpha}
\bibliography{Bibliography}

\end{document}